\documentclass[a4paper,12pt]{article}
\usepackage{}
\usepackage{amsmath}
\usepackage{color}
\usepackage{amsthm}

\title{\textbf{Marcinkiewicz's strong law of large numbers for non-additive expectation}}

\author{Lixin Zhang, {Jinghang Lin$^*$}\\
\small \emph {School of Mathematical Sciences, Zhejiang University, Hangzhou 310027, China}}

\date{\small \emph {Email addresses: stazlx@zju.edu.cn(L.-X.Zhang), staljh@zju.edu.cn(J.-H. Lin)}}
\makeatletter

\newtheorem{thm}{Theorem}
\newtheorem{deft}[thm]{Definition}
\newtheorem{lem}{Lemma}

\begin{document}
\maketitle

\begin{abstract}
The sub-linear expectation space is a nonlinear expectation space having advantages of modelling the uncertainty of probability and distribution. In the sub-linear expectation space, we use capacity and sub-linear expectation to replace probability and  expectation of classical probability theory. In this paper, the method of selecting subsequence is used to prove Marcinkiewicz type strong law of large numbers under sub-linear expectation space. This result is a natural extension of the classical Marcinkiewicz's strong law of large numbers to the case where the expectation is nonadditive. In addition, this paper also gives a theorem about convergence of a random series.

\textbf{keywords: strong law of large numbers,
capacity, nonlinear expectation }
\end{abstract}

\section{Introduction}
Under the frame of classical probability theory established by Kolmogorov, probability and expectation are both linear. To model uncertain phenomena in many areas, like finance, economics and insurance, sublinear expectation is introduced by Peng \cite{peng08a}. We use the notion of independence and identical distribution introduced by Peng. The main purpose of this paper is to establish Marcinkiewicz's strong law of large numbers under nonlinear expectation. Chen \cite{Chen13a} obtained Kolmogorov's strong law of large numbers for i.i.d. random variables under the condition of $1+\alpha $-moments. Zhang \cite{zhang16c} obtained the same result when the corresponding Choquet integral is finite. This can be regarded as a generalization of Kolomogorov's strong law of large numbers.

Our paper is organized as follows: we introduce some basic settings and notations in section 2. In section 3, under the frame introduced by Peng \cite{peng08a},we establish Marcinkiewicz's strong law of large numbers. Then, we introduce a theorem about convergence of random series.

\section{Basic Setting}
 We use notations of Peng \cite{peng08a}. Given a positive integer number \(n\), we denote by \(\left\langle \mathbf{x,y} \right\rangle \) the scalar product of \(\mathbf{x,y}\in {{R}^{n}}\) and by\(\left| \mathbf{x} \right|={{\left\langle \mathbf{x,x} \right\rangle}^{1/2}}\) the Euclidean norm of \(\mathbf{x}\).

Let \(\left( \Omega ,\mathcal {F} \right)\) be a measurable space and \(\mathcal {H}\) be a linear space of real functions defined on \(\left( \Omega ,\mathcal {F} \right)\) such that \(\varphi \left( {{X}_{1}},{{X}_{2}},\cdots ,{{X}_{n}} \right)\in \mathcal {H}\), for any \({{X}_{1}},{{X}_{2}},\cdots ,{{X}_{n}}\in \mathcal {H}\), \(\varphi \in {{C}_{l,Lip}}\left( {{R}^{n}} \right)\). where \({{C}_{l,Lip}}\left( {{R}^{n}} \right)\) denotes the linear space of local Lipschitz functions \(\varphi \) satisfying

\begin{center}
\(\left| \varphi \left( \mathbf{x} \right)-\varphi \left( \mathbf{y} \right) \right|\le C\left( 1+{{\left| \mathbf{x} \right|}^{m}}+{{\left| \mathbf{y} \right|}^{m}} \right)\left| \mathbf{x}-\mathbf{y} \right|,\ \forall \mathbf{x,y}\in {{R}^{n}},\)
\end{center}
for some \(C>0\), \(m\in N\) depending on \(\varphi \).

\(\mathcal {H}\) is considered as a space of "random variables".
\begin{deft}
 \(\hat{E}:\mathcal {H}\to \left( -\infty ,+\infty  \right)\) is called sublinear expectation, if \(\hat{E}\) satisfies the following properties: \(\forall X,Y\in \mathcal {H}\), we have
\begin{enumerate}
    \item Monotonicity: If \(X\ge Y\), then \(\hat{E}\left( X \right)\ge \hat{E}\left( Y \right)\).
    \item Constant preserving: \(\hat{E}\left( c \right)=c\).
    \item Sub-additivity: \(\hat{E}\left( X+Y \right)\le \hat{E}\left( X \right)+\hat{E}\left( Y \right)\).
    \item Positive homogeneity: \(\hat{E}\left( \lambda X \right)=\lambda \hat{E}\left( X \right), \lambda \ge 0\).
  \end{enumerate}
\end{deft}

The triple \((\Omega ,\mathcal {H},\hat{E} )\) is called a sub-linear expectation space, compared with the classical probability space \(\left( \Omega ,\mathcal {F},P \right)\). For convenience, we also call \(\hat{E}\) a nonlinear expectation.

In a sub-linear expectation space, we replace the concept of probability with the concept of capacity.
\\
\begin{deft}

A set function \(V:\mathcal {F}\to \left[0,1 \right]\) is called a capacity, if

1. \(V\left( \phi  \right)=0,V\left( \Omega  \right)=1\).

2. \(V\left( A \right)\le V\left( B \right),\forall A\subset B,A,B\in \mathcal {F}\).\\
In addition, if \(V\) is continuous,then \(V\) should also satisfy

3. \(V\left( {{A}_{n}} \right)\uparrow V\left( A \right)\),if \({{A}_{n}}\uparrow A\).

4. \(V\left( {{A}_{n}} \right)\downarrow V\left( A \right)\),if \({{A}_{n}}\downarrow A\), where \({{A}_{n}},A\in \mathcal {F}\).\\
A sub-linear expectation \(\hat{E}\) could generate  a pair of capacity denoted by \[\hat{V}\left( A \right)=\hat{E}\left( {{I}_{A}} \right), \hat{v}\left( A \right)=-\hat{E}\left( {{-I}_{A}} \right), if{{I}_{A}}\in \mathcal{H}.\]
We also have the fact
\begin{center}
  \(\hat{E}\left( f \right)\le V\left( A \right)\le \hat{E}\left( g \right),\hat{\varepsilon }\left( f \right)\le v\left( A \right)\le \hat{\varepsilon }\left( g \right),if \ f\le {{I}_{A}}\le g,f,g\in \mathcal{H}\)
\end{center}
It is easy to check that \(\hat{V}(A)+\hat{v}(A)=1.\)
\end{deft}
We define the Choquet integrals/expectations by\\
\centerline{\({{C}_{V}}\left( X \right)=\int_{0}^{\infty }{V\left( X\ge t \right)}dt+\int_{-\infty }^{0}{\left( V\left( X\ge t \right)-1 \right)}dt.\)}
where \(V\)can also be replaced by \(v\).

Under some conditions, the Borel-Cantelli lemma is still true for capacity.

\begin{lem}\cite{Chen13a}
Let \(\left\{ {{A}_{n}},n\ge 1 \right\}\) be a sequence of events in \(\mathcal {F}\). If \(V\) is continuous and \(\sum\limits_{n=1}^{\infty }{V\left( {{A}_{n}} \right)}<+\infty \), then \(V\left( \underset{n=1}{\overset{+\infty }{\mathop{\bigcap }}}\,\underset{i=n}{\overset{+\infty }{\mathop{\bigcup }}}\,{{A}_{i}} \right)=0\).
\end{lem}
\begin{proof}
By the monotonicity and continuity, it follows that\\
\(0\le V\left( \underset{n=1}{\overset{+\infty }{\mathop{\bigcap }}}\,\underset{i=n}{\overset{+\infty }{\mathop{\bigcup }}}\,{{A}_{i}} \right)\le V\left( \underset{i=n}{\overset{+\infty }{\mathop{\bigcup }}}\,{{A}_{i}} \right)\le \sum\limits_{i=n}^{\infty }{V\left( {{A}_{i}} \right)}\to 0\), if \(n\to \infty \).
\end{proof}

\begin{deft}\cite{peng08a}(Independence) In a sublinear expectation space \(( \Omega ,\mathcal {H},\hat{E} )\), for each test function \(\varphi \in {{C}_{l,Lip}}\left( {{R}^{m}}\times {{R}^{n}} \right)\),\(\mathbf{Y}=\left( {{Y}_{1}},\cdots ,{{Y}_{n}} \right),{{Y}_{i}}\in \mathcal {H}\) is said to be independent to \(\mathbf{X}=\left( {{X}_{1}},\cdots ,{{X}_{m}} \right), {{X}_{i}}\in \mathcal {H}\),
if we have \(\hat{E}\left( \varphi \left( \mathbf{X,Y} \right) \right)=\hat{E}\left( {{\left. \hat{E}\left( \varphi \left( \mathbf{x,Y} \right) \right) \right|}_{\mathbf{x=X}}} \right),\)
whenever
\(\bar{\varphi }\left( \mathbf{x} \right):=\hat{E}\left[ \left| \varphi \left(\mathbf{ x,Y} \right) \right| \right]<\infty \) for all \(\mathbf{x}\) and \(\hat{E}\left[ \left| \varphi \mathbf{X} \right| \right]<\infty \).

\end{deft}

 \begin{deft}\cite{peng08a}(Identical distribution)
  Let \(\mathbf{{{X}_{1}}}\) and \(\mathbf{{{X}_{2}}}\) be two \(n\)-dimensional random vectors defined respectively in sublinear expectation spaces \(({{\Omega }_{1}},{{\mathcal {H}}_{1}},{{{\hat{E}}}_{1}} )\) and \(( {{\Omega }_{2}},{{\mathcal {H}}_{2}},{{{\hat{E}}}_{2}} )\). \(\mathbf{{{X}_{1}}}\) and \(\mathbf{{{X}_{2}}}\) are called identically distributed, denoted by \(\mathbf{{{X}_{1}}}\overset{d}{\mathop{=}}\,\mathbf{{{X}_{2}}}\), if\\
\centerline{\({{\hat{E}}_{1}}\left( \varphi \left( \mathbf{{{X}_{1}}} \right) \right)={{\hat{E}}_{2}}\left( \varphi \left( \mathbf{{{X}_{2}}} \right) \right), \forall \varphi \in {{C}_{l,Lip}}\left( {{R}^{n}} \right),\)}
whenever the sub-expectation is finite.
\end{deft}

\begin{deft}\cite{peng08a}(IID random variables) A sequence of random variables \(\left\{ {{X}_{n}};n\ge 1 \right\}\)is said to be independent, if ${X}_{i+1}$ is independent to \(\left( {{X}_{1}},\cdots ,{{X}_{i}} \right)\) for each \(i\ge 1\). It is said to be identically distributed, if \({{X}_{i}}\overset{d}{=}\,{{X}_{1}}\) for each \(i\ge 1.\)

\end{deft}

\section{Main results}
In classical probability theory, the proof of Marcinkiewicz's strong law of large numbers depends on the three series theorem. In this paper, we use the method of picking up subsequence to prove it.
First, we prove a lemma.
\begin{lem}Let \({{C}_{V}}\left( {{\left| X \right|}^{p}} \right)<\infty ,1\le p<2\). \(X\) is a random variable defined on sublinear expectation space \(( \Omega ,\mathcal {H},\hat{E} )\). Then

\begin{center}
\(\sum\limits_{j=1}^{\infty }{\frac{\hat{E}\left( {{\left( \left| X \right|\wedge {{j}^{{1}/{p}\;}} \right)}^{2}} \right)}{{{j}^{{2}/{p}\;}}}}<\infty \).
\end{center}
\end{lem}

\begin{proof}
By noting that
\begin{align*}
  & {{\left( \left| X \right|\wedge {{j}^{{1}/{p}\;}} \right)}^{2}}=\sum\limits_{i=1}^{j}{{{\left| X \right|}^{2}}I\left( {{\left( i-1 \right)}^{{1}/{p}\;}}<\left| X \right|\le {{i}^{{1}/{p}\;}} \right)}+{{j}^{{2}/{p}\;}}I\left( \left| X \right|>{{j}^{{1}/{p}\;}} \right) \\
 & \ \ \ \ \ \ \ \ \ \ \ \ \ \ \ \ \ \ \le \sum\limits_{i=1}^{j}{{{i}^{{2}/{p}\;}}I\left( {{\left( i-1 \right)}^{{1}/{p}\;}}<\left| X \right|\le {{i}^{{1}/{p}\;}} \right)}+{{j}^{{2}/{p}\;}}I\left( \left| X \right|>{{j}^{{1}/{p}\;}} \right) \\
 & \ \ \ \ \ \ \ \ \ \ \ \ \ \ \ \ \ \ =\sum\limits_{i=0}^{j-1}{{{\left( i+1 \right)}^{{2}/{p}\;}}I\left( \left| X \right|>{{i}^{{1}/{p}\;}} \right)}-\sum\limits_{i=1}^{j}{{{i}^{{2}/{p}\;}}I\left( \left| X \right|>{{i}^{{1}/{p}\;}} \right)}+{{j}^{{2}/{p}\;}}I\left( \left| X \right|>{{j}^{{1}/{p}\;}} \right) \\
 & \ \ \ \ \ \ \ \ \ \ \ \ \ \ \ \ \ \ = 1+\sum\limits_{i=1}^{j-1}{\left( {{\left( i+1 \right)}^{{2}/{p}\;}}-{{i}^{{2}/{p}\;}} \right)}I\left( \left| X \right|>{{i}^{{1}/{p}\;}} \right)\\
 & \ \ \ \ \ \ \ \ \ \ \ \ \ \ \ \ \ \ \le 1+\sum\limits_{i=1}^{j}{\left( {{\left( i+1 \right)}^{{2}/{p}\;}}-{{i}^{{2}/{p}\;}} \right)}I\left( \left| X \right|>{{i}^{{1}/{p}\;}} \right),
\end{align*}
we have,
\begin{center}
\(\hat{E}\left( {{\left( X\wedge {{j}^{{1}/{p}\;}} \right)}^{2}} \right)\le \ 1+\sum\limits_{i=1}^{j}{\left( {{\left( i+1 \right)}^{{2}/{p}\;}}-{{i}^{{2}/{p}\;}} \right)}V\left( \left| X \right|>{{i}^{{1}/{p}\;}} \right)\).
\end{center}
Hence,
\begin{align*}
  & \sum\limits_{j=1}^{\infty }{\frac{\hat{E}\left( {{\left( X\wedge {{j}^{{1}/{p}\;}} \right)}^{2}} \right)}{{{j}^{{2}/{p}\;}}}}\le \ \sum\limits_{j=1}^{\infty }{\frac{1}{{{j}^{{2}/{p}\;}}}}+\sum\limits_{j=1}^{\infty }{{{j}^{{-2}/{p}\;}}\sum\limits_{i=1}^{j}{\left( {{\left( i+1 \right)}^{{2}/{p}\;}}-{{i}^{{2}/{p}\;}} \right)}V\left( \left| X \right|>{{i}^{{1}/{p}\;}} \right)} \\
 & \ \ \ \ \ \ \ \ \ \ \ \ \ \ \ \ \ \ \ \ \ \ \ \ \ \ \ \ \le \sum\limits_{j=1}^{\infty }{\frac{1}{{{j}^{{2}/{p}\;}}}}+\sum\limits_{i=1}^{\infty }{\left( {{\left( i+1 \right)}^{{2}/{p}\;}}-{{i}^{{2}/{p}\;}} \right)V\left( \left| X \right|>{{i}^{{1}/{p}\;}} \right)\sum\limits_{j=i}^{\infty }{{{j}^{{-2}/{p}\;}}}} \\
 & \ \ \ \ \ \ \ \ \ \ \ \ \ \ \ \ \ \ \ \ \ \ \ \ \ \ \ \ \le \sum\limits_{j=1}^{\infty }{\frac{1}{{{j}^{{2}/{p}\;}}}}+C\sum\limits_{i=1}^{\infty }{\frac{\left( {{\left( i+1 \right)}^{{2}/{p}\;}}-{{i}^{{2}/{p}\;}} \right)}{{{i}^{{2}/{p}\;-1}}}V\left( \left| X \right|>{{i}^{{1}/{p}\;}} \right)},
\end{align*}
where \(C\) is a positive constant.\\
Let \(f\left( x \right)=\frac{{{\left( x+1 \right)}^{{2}/{p}\;}}-{{x}^{{2}/{p}\;}}}{{{x}^{{2}/{p}\;-1}}}=x{{\left( 1+\frac{1}{x} \right)}^{{2}/{p}\;}}-x, x\ge 1\).

Since \({f}'\left( x \right)={{\left( 1+\frac{1}{x} \right)}^{{2}/{p}\;-1}}-1>0\),
\(f\left( x \right)\) is increasing.
And\(\underset{x\to \infty }{\mathop{\lim }}\,f\left( x \right)=\frac{2}{p}\), so,
\begin{align*}
  & \sum\limits_{j=1}^{\infty }{\frac{\hat{E}\left( {{\left( X\wedge {{j}^{{1}/{p}\;}} \right)}^{2}} \right)}{{{j}^{{2}/{p}\;}}}}\le \sum\limits_{j=1}^{\infty }{\frac{1}{{{j}^{{2}/{p}\;}}}}+C\sum\limits_{i=1}^{\infty }{\frac{\left( {{\left( i+1 \right)}^{{2}/{p}\;}}-{{i}^{{2}/{p}\;}} \right)}{{{i}^{{2}/{p}\;-1}}}V\left( \left| X \right|>{{i}^{{1}/{p}\;}} \right)} \\
 & \ \ \ \ \ \ \ \ \ \ \ \ \ \ \ \ \ \ \ \ \ \ \ \ \ \ \ \ \le \sum\limits_{j=1}^{\infty }{\frac{1}{{{j}^{{2}/{p}\;}}}}+C\cdot \frac{2}{p}\sum\limits_{i=1}^{\infty }{V\left( \left| X \right|>{{i}^{{1}/{p}\;}} \right)}\\
 & \ \ \ \ \ \ \ \ \ \ \ \ \ \ \ \ \ \ \ \ \ \ \ \ \ \ \ \ \le \sum\limits_{j=1}^{\infty }{\frac{1}{{{j}^{{2}/{p}\;}}}}+C\cdot \frac{2}{p}\cdot {{C}_{V}}\left( {{\left| X \right|}^{p}} \right)<\infty.
\end{align*}

\end{proof}
\begin{thm}\cite{zhang16c} Suppose that \({X}_{1},...,{{X}_{n}}\) is a sequence of independent random variable and
\(\hat{E}\left( {{X}_{k}} \right)\le 0,k=1,\cdots ,n.\) Then
\begin{center}
\(\hat{E}\left( {{\left| \underset{k\le n}{\mathop{\max }}\,\left( {{S}_{n}}-{{S}_{k}} \right) \right|}^{p}} \right)\le {{2}^{2-p}}\sum\limits_{k=1}^{n}{\hat{E}\left( {{\left| {{X}_{k}} \right|}^{p}} \right)},1\le p\le 2,\)
\end{center}
where \({{S}_{k}}={{X}_{1}}+\cdots +{{X}_{k}}.\)

In particular, \(\hat{E}\left( {{\left( S_{n}^{+} \right)}^{p}} \right)\le {{2}^{2-p}}\sum\limits_{k=1}^{n}{\hat{E}\left( {{\left| {{X}_{k}} \right|}^{p}} \right)},1\le p\le 2.\)

\end{thm}

By Lemma 2, we can prove Marcinkiewz's strong law of large numbers.

\begin{thm}
In sublinear expectation space \(( \Omega ,\mathcal {H},\hat{E} )\), \(\left\{ {{X}_{i}} \right\}\) is a sequence of independent and identical random variable. Suppose that \(1\le p<2\), \(\hat{E}\left( {{X}_{i}} \right)=\hat{\varepsilon }\left( {{X}_{i}} \right)=0\), \(\underset{c\to \infty }{\mathop{\lim }}\,\hat{E}\left[ {{\left( {{\left| {{X}_{1}} \right|}^{p}}-c \right)}^{+}} \right]=0\) and $V$ is continuous, then \({{C}_{V}}\left( {{\left| {{X}_{1}} \right|}^{p}} \right)<\infty \) if and only if \({{{S}_{n}}}/{{{n}^{1/p}}}\;\to 0\ \ a.s.\ V.\)
\end{thm}

\begin{proof}
\(\Rightarrow \) Define \({{Y}_{k}}=\left( -{{k}^{{1}/{p}\;}} \right)\vee \left( {{X}_{k}}\wedge {{k}^{{1}/{p}\;}} \right).\)\\
So,

\(\sum\limits_{k=1}^{\infty }{V\left( {{Y}_{k}}\ne {{X}_{k}} \right)}=\sum\limits_{k=1}^{\infty }{V\left( \left| {{X}_{1}} \right|>{{k}^{1/p}} \right)}\le \int_{0}^{\infty }{V\left( \left| {{X}_{1}} \right|>{{t}^{1/p}} \right)dt={{C}_{V}}\left( {{\left| {{X}_{1}} \right|}^{p}} \right)<\infty }\).\\
By the Borel-Cantelli lemma, we can get \(V\left( {{X}_{k}}\ne {{Y}_{k}},i.o. \right)=0\).
That is
\\
\centerline{\(v\left( \underset{k\to \infty }{\mathop{\lim \inf }}\,\left( {{X}_{k}}={{Y}_{k}} \right) \right)=1.\)}
\(\forall \omega \in \underset{k\to \infty }{\mathop{\lim \inf }}\,\left( {{X}_{k}}={{Y}_{k}} \right)\), \(\exists K\), s.t.\\
\centerline{\({{X}_{k}}\left( \omega  \right)={{Y}_{k}}\left( \omega  \right)\), when \(k>K\).}
 so \(\sum\limits_{k\ge K}{{{X}_{k}}}=\sum\limits_{k\ge K}{{{Y}_{k}}}\).

In order to prove \({{{S}_{n}}}/{{{n}^{1/p}}}\;\to 0\ a.s.\ V\), we only need to prove \({{{{\bar{S}}}_{n}}}/{{{n}^{1/p}}}\;\to 0\ a.s.\ V\), where
\({{\bar{S}}_{n}}=\sum\limits_{k=1}^{n}{{{Y}_{k}}}\).
Let \({{n}_{k}}={{[\theta]}^{k}}\), \(\theta >1\).

Step one: we first prove
\(\frac{{{{\bar{S}}}_{{{n}_{k}}}}}{n_{k}^{{1}/{p}\;}}\to 0\ \ \ a.s.\ V.\)\\
\centerline{\(V\left( \frac{{{{\bar{S}}}_{{{n}_{k}}}}}{n_{k}^{{1}/{p}\;}}\ge \varepsilon  \right)=V\left( \sum\limits_{j=1}^{{{n}_{k}}}{\left( {{Y}_{j}}-\hat{E}{{Y}_{j}} \right)}\ge \varepsilon \cdot n_{k}^{{1}/{p}\;}-\sum\limits_{j=1}^{{{n}_{k}}}{\hat{E}{{Y}_{j}}} \right).\)}
Now we prove \(\frac{1}{2}\varepsilon \cdot n_{k}^{{1}/{p}\;}-\sum\limits_{j=1}^{{{n}_{k}}}{\hat{E}{{Y}_{j}}\ge 0}\), i.e.
\(\frac{2\sum\limits_{j=1}^{{{n}_{k}}}{\hat{E}{{Y}_{j}}}}{{{\left( {{n}_{k}} \right)}^{{1}/{p}\;}}}\le \varepsilon .\)\\
Since \(\hat{E}\left( {{X}_{j}} \right)=0\), we have
\begin{align*}
  & \left| \hat{E}{{Y}_{j}} \right|=\left| \hat{E}\left( {{X}_{j}}I\left( \left| {{X}_{j}} \right|\le {{j}^{{1}/{p}\;}} \right)+{{j}^{{1}/{p}\;}}I\left( \left| {{X}_{j}} \right|>{{j}^{{1}/{p}\;}} \right) \right) \right| \\
 & \ \ \ \ \ \ \ \ \le \left| \hat{E}\left( {{X}_{j}}I\left( \left| {{X}_{j}} \right|\le {{j}^{{1}/{p}\;}} \right) \right) \right|+\left| \hat{E}\left( {{j}^{{1}/{p}\;}}I\left( \left| {{X}_{j}} \right|>{{j}^{{1}/{p}\;}} \right) \right) \right| \\
 & \ \ \ \ \ \ \ \ \le \left| \hat{E}\left( {{X}_{j}}\left( 1-I\left( \left| {{X}_{j}} \right|>{{j}^{{1}/{p}\;}} \right) \right) \right) \right|+\left| \hat{E}\left( \left| {{X}_{j}} \right|I\left( \left| {{X}_{j}} \right|>{{j}^{{1}/{p}\;}} \right) \right) \right| \\
 & \ \ \ \ \ \ \ \ \le 2\cdot \left| \hat{E}\left( \left| {{X}_{j}} \right|I\left( \left| {{X}_{j}} \right|>{{j}^{{1}/{p}\;}} \right) \right) \right|=2{{j}^{{1}/{p}\;}}\cdot \left| \hat{E}\left( \frac{\left| {{X}_{j}} \right|}{{{j}^{{1}/{p}\;}}}I\left( \left| {{X}_{j}} \right|>{{j}^{{1}/{p}\;}} \right) \right) \right| \\
 & \ \ \ \ \ \ \ \ \le 2{{j}^{\left( {1}/{p-1}\; \right)}}\cdot \left| \hat{E}\left( {{\left| {{X}_{1}} \right|}^{p}}I\left( {{\left| {{X}_{1}} \right|}^{p}}>j \right) \right) \right|
\end{align*}
Hence,
\begin{center}
\(\sum\limits_{j=1}^{{{n}_{k}}}{\left| \hat{E}{{Y}_{j}} \right|}\le \sum\limits_{j=1}^{{{n}_{k}}}{{2{j}^{\left( {1}/{p}\;-1 \right)}}\hat{E}\left( {{\left| {{X}_{1}} \right|}^{p}}I\left( {{\left| {{X}_{1}} \right|}^{p}}>j \right) \right)}\).
\end{center}
Since \(\underset{c\to \infty }{\mathop{\lim }}\,\hat{E}\left[ {{\left( {{\left| {{X}_{1}} \right|}^{p}}-c \right)}^{+}} \right]=0\), then\\
 \centerline{\(\hat{E}\left( {{\left| {{X}_{1}} \right|}^{p}}I\left( {{\left| {{X}_{1}} \right|}^{p}}>j \right) \right)\to 0\), when \(j\to \infty \) .}
Therefore,
\begin{center}
\(\frac{\sum\limits_{j=1}^{{{n}_{k}}}{{{j}^{\left( {1}/{p}\;-1 \right)}}\hat{E}\left( {{\left| {{X}_{1}} \right|}^{p}}I\left( {{\left| {{X}_{1}} \right|}^{p}}>j \right) \right)}}{{{\left( {{n}_{k}} \right)}^{{1}/{p}\;}}}\to 0\), when \(k\to \infty \).
\end{center}
We can get \(\frac{1}{2}\varepsilon \cdot n_{k}^{{1}/{p}\;}-\sum\limits_{j=1}^{{{n}_{k}}}{\left| \hat{E}{{Y}_{j}} \right|\ge 0}\). Furthermore,
\(\frac{1}{2}\varepsilon \cdot n_{k}^{{1}/{p}\;}-\sum\limits_{j=1}^{{{n}_{k}}}{\hat{E}{{Y}_{j}}\ge 0}.\)\\
It's obvious that \(\hat{E}\left( {{Y}_{j}}-\hat{E}{{Y}_{j}} \right)\le 0\).\\
By theorem 5,
\begin{align*}
  & V\left( \frac{{{{\bar{S}}}_{{{n}_{k}}}}}{n_{k}^{{1}/{p}\;}}\ge \varepsilon  \right)=V\left( \sum\limits_{j=1}^{{{n}_{k}}}{\left( {{Y}_{j}}-\hat{E}{{Y}_{j}} \right)}\ge \varepsilon \cdot n_{k}^{{1}/{p}\;}-\sum\limits_{j=1}^{{{n}_{k}}}{\hat{E}{{Y}_{j}}} \right) \\
 & \ \ \ \ \ \ \ \ \ \ \ \ \ \ \ \ \ \ \ \le V\left( \sum\limits_{j=1}^{{{n}_{k}}}{\left( {{Y}_{j}}-\hat{E}{{Y}_{j}} \right)}\ge \frac{1}{2}\varepsilon \cdot n_{k}^{{1}/{p}\;} \right)\le V\left( \sum\limits_{j=1}^{{{n}_{k}}}{{\left( {{Y}_{j}}-\hat{E}{{Y}_{j}} \right)}}^{+}\ge \frac{1}{2}\varepsilon \cdot n_{k}^{{1}/{p}\;} \right)\\
 & \ \ \ \ \ \ \ \ \ \ \ \ \ \ \ \ \ \ \ \le \frac{\hat{E}\left( {{\left( \sum\limits_{j=1}^{{{n}_{k}}}{{\left( {{Y}_{j}}-\hat{E}{{Y}_{j}} \right)}}^{+} \right)}^{2}} \right)}{{{\left( \frac{1}{2}\varepsilon \cdot n_{k}^{{1}/{p}\;} \right)}^{2}}}\le \frac{\hat{E}\left( {{\left| \underset{1\le j\le {{n}_{k}}}{\mathop{\max }}\,\sum\limits_{i=1}^{j}{{\left( {{Y}_{i}}-\hat{E}{{Y}_{i}} \right)}}^{+} \right|}^{2}} \right)}{{{\left( \frac{1}{2}\varepsilon \cdot n_{k}^{{1}/{p}\;} \right)}^{2}}} \\
 & \ \ \ \ \ \ \ \ \ \ \ \ \ \ \ \ \ \ \ \le \frac{\sum\limits_{j=1}^{{{n}_{k}}}{\hat{E}\left( {{\left| {{Y}_{j}}-\hat{E}{{Y}_{j}} \right|}^{2}} \right)}}{{{\left( \frac{1}{2}\varepsilon \cdot n_{k}^{{1}/{p}\;} \right)}^{2}}}\le \frac{4\sum\limits_{j=1}^{{{n}_{k}}}{\hat{E}\left( {{\left| {{Y}_{j}} \right|}^{2}} \right)}}{{{\left( \frac{1}{2}\varepsilon \cdot n_{k}^{{1}/{p}\;} \right)}^{2}}}.
\end{align*}
Since
\begin{align*}
  & \sum\limits_{k=1}^{\infty }{\frac{1}{n_{k}^{{2}/{p}\;}}\cdot \sum\limits_{j=1}^{{{n}_{k}}}{\hat{E}\left( {{\left| {{Y}_{j}} \right|}^{2}} \right)}}=\sum\limits_{j=1}^{\infty }{\hat{E}\left( {{\left| {{Y}_{j}} \right|}^{2}} \right)\cdot \sum\limits_{k:{{n}_{k}}\ge j}{\frac{1}{n_{k}^{{2}/{p}\;}}}} \\
 & \ \ \ \ \ \ \ \ \ \ \ \ \ \ \ \ \ \ \ \ \ \ \ \ \ \ \ \ \ \ \le \sum\limits_{j=1}^{\infty }{\frac{\hat{E}\left( {{\left| {{Y}_{j}} \right|}^{2}} \right)}{{{j}^{{2}/{p}\;}}}}\le \sum\limits_{j=1}^{\infty }{\frac{\hat{E}\left( {{\left| {{X}_{j}}\wedge {{j}^{{1}/{p}\;}} \right|}^{2}} \right)}{{{j}^{{2}/{p}\;}}}}<\infty,
\end{align*}
we have\\
\centerline{\(\sum\limits_{k=1}^{\infty }{V\left( \frac{{{{\bar{S}}}_{{{n}_{k}}}}}{n_{k}^{{1}/{p}\;}}\ge \varepsilon  \right)}<\infty \).}\\
By the Borel-Cantelli lemma, we can get
\(V\left( \frac{{{{\bar{S}}}_{{{n}_{k}}}}}{n_{k}^{{1}/{p}\;}}\ge \varepsilon ,i.o. \right)=0\).\\
Considering \(\left\{ -{{X}_{i}} \right\}\), by the same way, we can get\\
\centerline{\(V\left( \frac{-{{{\bar{S}}}_{{{n}_{k}}}}}{n_{k}^{{1}/{p}\;}}\ge \varepsilon ,i.o. \right)=0\).}
By noting that,\\
\centerline{\(V\left( \frac{\left| {{{\bar{S}}}_{{{n}_{k}}}} \right|}{n_{k}^{{1}/{p}\;}}\ge \varepsilon ,i.o. \right)\le V\left( \frac{{{{\bar{S}}}_{{{n}_{k}}}}}{n_{k}^{{1}/{p}\;}}\ge \varepsilon ,i.o. \right)+V\left( \frac{-{{{\bar{S}}}_{{{n}_{k}}}}}{n_{k}^{{1}/{p}\;}}\ge \varepsilon ,i.o. \right)=0\).}
Hence,\\
\centerline{\(\frac{{{{\bar{S}}}_{{{n}_{k}}}}}{n_{k}^{{1}/{p}\;}}\to 0\ \ \ a.s.\ V.\)}
\newpage
Step two: Next, we will prove \(\frac{\underset{{{n}_{k-1}}< n\le {{n}_{k}}}{\mathop{\max }}\,\left| {{{\bar{S}}}_{n}}-{{{\bar{S}}}_{{{n}_{k}}}} \right|}{n_{k}^{{1}/{p}\;}}\to 0\ a.s.V.\)\\
We only prove \(\frac{\underset{{{n}_{k-1}}< n\le {{n}_{k}}}{\mathop{\max }}\,{{\left( {{{\bar{S}}}_{{{n}_{k}}}}-{{{\bar{S}}}_{n}} \right)}^{+}}}{n_{k}^{{1}/{p}\;}}\to 0\ a.s.V\),
because the proof of \(\frac{\underset{{{n}_{k-1}}< n\le {{n}_{k}}}{\mathop{\max }}\,{{\left( {{{\bar{S}}}_{{{n}_{k}}}}-{{{\bar{S}}}_{n}} \right)}^{-}}}{n_{k}^{{1}/{p}\;}}\to 0\ a.s.V\) is similar.
\begin{align*}
  & V\left( \underset{{{n}_{k-1}}< n\le {{n}_{k}}}{\mathop{\max }}\,{{\left( {{{\bar{S}}}_{{{n}_{k}}}}-{{{\bar{S}}}_{n}} \right)}^{+}}\ge \varepsilon \cdot n_{k}^{{1}/{p}\;} \right) \\
 & \le V\left( \underset{{{n}_{k-1}}< n\le {{n}_{k}}}{\mathop{\max }}\,{{\left( {{{\bar{S}}}_{{{n}_{k}}}}-{{{\bar{S}}}_{n}}-\sum\limits_{j={n+1}}^{{n}_{k}}{\hat{E}{{Y}_{j}}} \right)}^{+}}\ge \varepsilon \cdot n_{k}^{{1}/{p}\;}-\sum\limits_{j={{n}_{k-1}+1}}^{{{n}_{k}}}{\left| \hat{E}{{Y}_{j}} \right|} \right).(*)
\end{align*}
Now, we prove \(\frac{1}{2}\varepsilon \cdot n_{k}^{{1}/{p}\;}-\sum\limits_{j={{n}_{k-1}+1}}^{{{n}_{k}}}{\left| \hat{E}{{Y}_{j}} \right|\ge 0}\), i.e.
\(\frac{2\sum\limits_{j={{n}_{k-1}+1}}^{{{n}_{k}}}{\left| \hat{E}{{Y}_{j}} \right|}}{{{\left( {{n}_{k}} \right)}^{{1}/{p}\;}}}\le \varepsilon .\)\\
Since \(\hat{E}\left( {{X}_{j}} \right)=0\), we can get
\begin{align*}
  & \left| \hat{E}{{Y}_{j}} \right|=\left| \hat{E}\left( {{X}_{j}}I\left( \left| {{X}_{j}} \right|\le {{j}^{{1}/{p}\;}} \right)+{{j}^{{1}/{p}\;}}I\left( \left| {{X}_{j}} \right|>{{j}^{{1}/{p}\;}} \right) \right) \right| \\
 & \ \ \ \ \ \ \ \ \le \left| \hat{E}\left( {{X}_{j}}I\left( \left| {{X}_{j}} \right|\le {{j}^{{1}/{p}\;}} \right) \right) \right|+\left| \hat{E}\left( {{j}^{{1}/{p}\;}}I\left( \left| {{X}_{j}} \right|>{{j}^{{1}/{p}\;}} \right) \right) \right| \\
 & \ \ \ \ \ \ \ \ \le \left| \hat{E}\left( {{X}_{j}}\left( 1-I\left( \left| {{X}_{j}} \right|>{{j}^{{1}/{p}\;}} \right) \right) \right) \right|+\left| \hat{E}\left( \left| {{X}_{j}} \right|I\left( \left| {{X}_{j}} \right|>{{j}^{{1}/{p}\;}} \right) \right) \right| \\
 & \ \ \ \ \ \ \ \ \le 2\cdot \left| \hat{E}\left( \left| {{X}_{j}} \right|I\left( \left| {{X}_{j}} \right|>{{j}^{{1}/{p}\;}} \right) \right) \right|=2{{j}^{{1}/{p}\;}}\cdot \left| \hat{E}\left( \frac{\left| {{X}_{j}} \right|}{{{j}^{{1}/{p}\;}}}I\left( \left| {{X}_{j}} \right|>{{j}^{{1}/{p}\;}} \right) \right) \right| \\
 & \ \ \ \ \ \ \ \ \le 2{{j}^{\left( {1}/{p-1}\; \right)}}\cdot \left| \hat{E}\left( {{\left| {{X}_{1}} \right|}^{p}}I\left( {{\left| {{X}_{1}} \right|}^{p}}>j \right) \right) \right|
\end{align*}
So,
\begin{center}
\(\sum\limits_{j={{n}_{k-1}+1}}^{{{n}_{k}}}{\left| \hat{E}{{Y}_{j}} \right|}\le \sum\limits_{j={{n}_{k-1}+1}}^{{{n}_{k}}}{{2{j}^{\left( {1}/{p}\;-1 \right)}}\hat{E}\left( {{\left| {{X}_{1}} \right|}^{p}}I\left( {{\left| {{X}_{1}} \right|}^{p}}>j \right) \right)}\).
\end{center}
Since \(\underset{c\to \infty }{\mathop{\lim }}\,\hat{E}\left[ {{\left( {{\left| {{X}_{1}} \right|}^{p}}-c \right)}^{+}} \right]=0\), then\\
 \centerline{\(\hat{E}\left( {{\left| {{X}_{1}} \right|}^{p}}I\left( {{\left| {{X}_{1}} \right|}^{p}}>j \right) \right)\to 0\), when \(j\to \infty \) .}
Therefore,
\begin{center}
\(\frac{\sum\limits_{j={{n}_{k-1}+1}}^{{{n}_{k}}}{{{j}^{\left( {1}/{p}\;-1 \right)}}\hat{E}\left( {{\left| {{X}_{1}} \right|}^{p}}I\left( {{\left| {{X}_{1}} \right|}^{p}}>j \right) \right)}}{{{\left( {{n}_{k}} \right)}^{{1}/{p}\;}}}\to 0\), when \(k\to \infty \).
\end{center}
We can get \(\frac{1}{2}\varepsilon \cdot n_{k}^{{1}/{p}\;}-\sum\limits_{j={{n}_{k-1}+1}}^{{{n}_{k}}}{\left| \hat{E}{{Y}_{j}} \right|\ge 0}\).\\
since \(\hat{E}\left( {{Y}_{i}}-\hat{E}\left( {{Y}_{i}} \right) \right)\le 0\),
(*) can be transformed into
\begin{align*}
  & V\left( \underset{{{n}_{k-1}}<n\le {{n}_{k}}}{\mathop{\max }}\,{{\left( {{{\bar{S}}}_{{{n}_{k}}}}-{{{\bar{S}}}_{n}} \right)}^{+}}\ge \varepsilon \cdot n_{k}^{{1}/{p}\;} \right)\\
  & \le V\left( \underset{{{n}_{k-1}}<n\le {{n}_{k}}}{\mathop{\max }}\,{{\left( {{{\bar{S}}}_{{{n}_{k}}}}-{{{\bar{S}}}_{n}}-\sum\limits_{j=n+1}^{{{n}_{k}}}{\hat{E}{{Y}_{j}}} \right)}^{+}}\ge \frac{1}{2}\varepsilon  n_{k}^{{1}/{p}\;} \right) \\
 &  =V\left( \underset{{{n}_{k-1}}<n\le {{n}_{k}}}{\mathop{\max }}\,{{\left( \sum\limits_{j=n+1}^{{{n}_{k}}}{\left( {{Y}_{j}}-\hat{E}{{Y}_{j}} \right)} \right)}^{+}}\ge \frac{1}{2}\varepsilon  n_{k}^{{1}/{p}\;} \right) \\
 &  \le \frac{\hat{E}\left( {{\left( \underset{{{n}_{k-1}}<n\le {{n}_{k}}}{\mathop{\max }}\,{{\left( \sum\limits_{j=n+1}^{{{n}_{k}}}{\left( {{Y}_{j}}-\hat{E}{{Y}_{j}} \right)} \right)}^{+}} \right)}^{2}} \right)}{{{\left( \frac{1}{2}\varepsilon \cdot n_{k}^{{1}/{p}\;} \right)}^{2}}} \\
 &  \le \frac{\sum\limits_{j={{n}_{k-1}}+1}^{{{n}_{k}}}{\hat{E}\left[ {{\left| {{Y}_{j}}-\hat{E}{{Y}_{j}} \right|}^{2}} \right]}}{{{\left( \frac{1}{2}\varepsilon \cdot n_{k}^{{1}/{p}\;} \right)}^{2}}}\le \frac{4\sum\limits_{j={{n}_{k-1}}+1}^{{{n}_{k}}}{\hat{E}\left[ {{\left| {{Y}_{j}} \right|}^{2}} \right]}}{{{\left( \frac{1}{2}\varepsilon \cdot n_{k}^{{1}/{p}\;} \right)}^{2}}}.
\end{align*}
So,\\
\begin{align*}
  & \sum\limits_{k=1}^{\infty }{\frac{\sum\limits_{j={{n}_{k-1}}\text{+}1}^{{{n}_{k}}}{\hat{E}\left[ {{\left| {{Y}_{j}} \right|}^{2}} \right]}}{n_{k}^{{2}/{p}\;}}}\le \sum\limits_{k=1}^{\infty }{\sum\limits_{j={{n}_{k-1}}\text{+}1}^{{{n}_{k}}}{\frac{\hat{E}\left[ {{\left| {{Y}_{j}} \right|}^{2}} \right]}{{{j}^{{2}/{p}\;}}}}}=\sum\limits_{j=2}^{\infty }{\frac{\hat{E}\left[ {{\left| {{Y}_{j}} \right|}^{2}} \right]}{{{j}^{{2}/{p}\;}}}} \\
 & \ \ \ \ \ \ \ \ \ \ \ \ \ \ \ \ \ \ \ \ \ \ \ \ \ \ \ \ <\sum\limits_{j=1}^{\infty }{\frac{\hat{E}\left[ {{\left| {{X}_{j}}\wedge {{j}^{{1}/{p}\;}} \right|}^{2}} \right]}{{{j}^{{2}/{p}\;}}}}<\infty .
\end{align*}

By Borel-Cantelli lemma, we can get\\
\centerline{\(V\left( \frac{\underset{{{n}_{k-1}}< n\le {{n}_{k}}}{\mathop{\max }}\,{{\left( {{{\bar{S}}}_{{{n}_{k}}}}-{{{\bar{S}}}_{n}} \right)}^{+}}}{n_{k}^{{1}/{p}\;}}\ge \varepsilon ,i.o. \right)=0\).}\\
Hence,
\begin{center}
\(\frac{\underset{{{n}_{k-1}}< n\le {{n}_{k}}}{\mathop{\max }}\,{{\left( {{{\bar{S}}}_{{{n}_{k}}}}-{{{\bar{S}}}_{n}} \right)}^{+}}}{n_{k}^{{1}/{p}\;}}\to 0\ a.s.V.\)
\end{center}
Noting that \(\left| {{{\bar{S}}}_{{{n}_{k}}}}-{{{\bar{S}}}_{n}} \right|={{\left( {{{\bar{S}}}_{{{n}_{k}}}}-{{{\bar{S}}}_{n}} \right)}^{+}}+{{\left( {{{\bar{S}}}_{{{n}_{k}}}}-{{{\bar{S}}}_{n}} \right)}^{-}}\), finally we can get,\\
\centerline{\(\frac{\underset{{{n}_{k-1}}< n\le {{n}_{k}}}{\mathop{\max }}\,\left| {{{\bar{S}}}_{n}}-{{{\bar{S}}}_{{{n}_{k}}}} \right|}{n_{k}^{{1}/{p}\;}}\to 0\ a.s.V.\)}

\(\Leftarrow \)
Suppose \({{C}_{V}}\left( {{\left| {{X}_{1}} \right|}^{p}} \right)=\infty \).\\
Let \({{g}_{\varepsilon }}\) be a function satisfying that its derivatives of each order are bounded, \({{g}_{\varepsilon }}\left( x \right)=1\) if \(x>1\), \({{g}_{\varepsilon }}\left( x \right)=0\) if \(x\le 1-\varepsilon \), and \(0\le {{g}_{\varepsilon }}\left( x \right)\le 1\) for all \(x\), where \(0<\varepsilon <1\). Then\\
 \centerline{\({{g}_{\varepsilon }}\left( \cdot  \right)\in {{C}_{l,Lip}}\left( R \right)\) and \(I\left\{ x\ge 1 \right\}\le {{g}_{\varepsilon }}\left( x \right)\le I\left\{ x>1-\varepsilon  \right\}\).}\\
 So,
\begin{align*}
  & \sum\limits_{j=1}^{\infty }{\hat{E}\left( {{g}_{{1}/{2}\;}}\left( \frac{\left| {{X}_{j}} \right|}{{{\left( Mj \right)}^{{1}/{p}\;}}} \right) \right)}=\sum\limits_{j=1}^{\infty }{\hat{E}\left( {{g}_{{1}/{2}\;}}\left( \frac{\left| {{X}_{1}} \right|}{{{\left( Mj \right)}^{{1}/{p}\;}}} \right) \right)} \\
 & \ \ \ \ \ \ \ \ \ \ \ \ \ \ \ \ \ \ \ \ \ \ \ \ \ \ \ \ \ \ \ \ \ \ \ \ge \sum\limits_{j=1}^{\infty }{V\left( \left| {{X}_{1}} \right|>{{\left( Mj \right)}^{{1}/{p}\;}} \right)}=\infty ,\ \ \forall M>0.(**)
\end{align*}
For any \(l\ge 1\),
\begin{align*}
  & v\left( \sum\limits_{j=1}^{n}{{{g}_{{1}/{2}\;}}\left( \frac{\left| {{X}_{1}} \right|}{{{\left( Mj \right)}^{{1}/{p}\;}}} \right)<l} \right)=v\left( \exp \left\{ -\frac{1}{2}\sum\limits_{j=1}^{n}{{{g}_{{1}/{2}\;}}\left( \frac{\left| {{X}_{1}} \right|}{{{\left( Mj \right)}^{{1}/{p}\;}}} \right)} \right\}>\exp \left( -\frac{l}{2} \right) \right) \\
 & \ \ \ \ \ \ \ \ \ \ \ \ \ \ \ \ \ \ \ \ \ \ \ \ \ \ \ \ \ \ \ \ \ \ \ \ \ \ \ \le {{e}^{{l}/{2}\;}}\hat{\varepsilon }\left( \exp \left\{ -\frac{1}{2}\sum\limits_{j=1}^{n}{{{g}_{{1}/{2}\;}}\left( \frac{\left| {{X}_{1}} \right|}{{{\left( Mj \right)}^{{1}/{p}\;}}} \right)} \right\} \right) \\
 & \ \ \ \ \ \ \ \ \ \ \ \ \ \ \ \ \ \ \ \ \ \ \ \ \ \ \ \ \ \ \ \ \ \ \ \  \ \ \ \le {{e}^{{l}/{2}\;}}\prod\limits_{j=1}^{n}{\hat{\varepsilon }\left( \exp \left\{ -\frac{1}{2}{{g}_{{1}/{2}\;}}\left( \frac{\left| {{X}_{1}} \right|}{{{\left( Mj \right)}^{{1}/{p}\;}}} \right) \right\} \right)}
\end{align*}
by the independence and \(0\le \exp \left\{ -\frac{1}{2}\sum\limits_{j=1}^{n}{{{g}_{{1}/{2}\;}}\left( \frac{\left| {{X}_{1}} \right|}{{{\left( Mj \right)}^{{1}/{p}\;}}} \right)} \right\}\in {{C}_{l,Lip}}\left( R \right)\).\\
Applying the elementary inequality\\
\centerline{\({{e}^{-x}}\le 1-\frac{1}{2}x\le {{e}^{-{x}/{2}\;}},\ \forall 0\le x\le \frac{1}{2}\)}
yields\\
\begin{center}
\(\exp \left\{ -\frac{1}{2}{{g}_{{1}/{2}\;}}\left( \frac{\left| {{X}_{1}} \right|}{{{\left( Mj \right)}^{{1}/{p}\;}}} \right) \right\}\le 1-\frac{1}{4}{{g}_{{1}/{2}\;}}\left( \frac{\left| {{X}_{1}} \right|}{{{\left( Mj \right)}^{{1}/{p}\;}}} \right)\le \exp \left\{ -\frac{1}{4}{{g}_{{1}/{2}\;}}\left( \frac{\left| {{X}_{1}} \right|}{{{\left( Mj \right)}^{{1}/{p}\;}}} \right) \right\}\).
\end{center}
It follows that\\
\begin{center}
 \(\hat{\varepsilon }\left( \exp \left\{ -\frac{1}{2}{{g}_{{1}/{2}\;}}\left( \frac{\left| {{X}_{1}} \right|}{{{\left( Mj \right)}^{{1}/{p}\;}}} \right) \right\} \right)\le \hat{\varepsilon }\left( \exp \left\{ -\frac{1}{4}{{g}_{{1}/{2}\;}}\left( \frac{\left| {{X}_{1}} \right|}{{{\left( Mj \right)}^{{1}/{p}\;}}} \right) \right\} \right)\).
\end{center}
since\\
\centerline{\(\hat{\varepsilon }\left( \exp \left\{ -\frac{1}{4}{{g}_{{1}/{2}\;}}\left( \frac{\left| {{X}_{1}} \right|}{{{\left( Mj \right)}^{{1}/{p}\;}}} \right) \right\} \right)\le \exp \left\{ \hat{\varepsilon }\left( -\frac{1}{4}{{g}_{{1}/{2}\;}}\left( \frac{\left| {{X}_{1}} \right|}{{{\left( Mj \right)}^{{1}/{p}\;}}} \right) \right) \right\}\) and}
\centerline{\(\exp \left\{ \hat{\varepsilon }\left( -\frac{1}{4}{{g}_{{1}/{2}\;}}\left( \frac{\left| {{X}_{1}} \right|}{{{\left( Mj \right)}^{{1}/{p}\;}}} \right) \right) \right\}=\exp \left\{ -\frac{1}{4}\hat{E}\left( {{g}_{{1}/{2}\;}}\left( \frac{\left| {{X}_{1}} \right|}{{{\left( Mj \right)}^{{1}/{p}\;}}} \right) \right) \right\}\),}
we have\\
\centerline{\(v\left( \sum\limits_{j=1}^{n}{{{g}_{{1}/{2}\;}}\left( \frac{\left| {{X}_{1}} \right|}{{{\left( Mj \right)}^{{1}/{p}\;}}} \right)<l} \right)\le {{e}^{\frac{l}{2}}}\exp \left\{ -\frac{1}{4}\hat{E}\left( {{g}_{{1}/{2}\;}}\left( \frac{\left| {{X}_{1}} \right|}{{{\left( Mj \right)}^{{1}/{p}\;}}} \right) \right) \right\}\to 0,\) as \(n\to \infty \).}
By(**),
we have\\
\begin{center}
\(V\left( \sum\limits_{j=1}^{n}{{{g}_{{1}/{2}\;}}\left( \frac{\left| {{X}_{1}} \right|}{{{\left( Mj \right)}^{{1}/{p}\;}}} \right)>l} \right)\to 1\), as \(n\to \infty \).
\end{center}
By continuity of \(V\), we can get
\begin{align*}
  & V\left( \underset{n\to \infty }{\mathop{\lim \sup }}\,\frac{\left| {{X}_{n}} \right|}{{{n}^{{1}/{p}\;}}}>\frac{M}{2} \right)=V\left( \frac{\left| {{X}_{j}} \right|}{{{\left( Mj \right)}^{{1}/{p}\;}}},i.o \right) \\
 & \ \ \ \ \ \ \ \ \ \ \ \ \ \ \ \ \ \ \ \ \ \ \ \ \ \ \ \ \ \ \ \ge V\left( \sum\limits_{j=1}^{\infty }{{{g}_{{1}/{2}\;}}\left( \frac{\left| {{X}_{j}} \right|}{{{\left( Mj \right)}^{{1}/{p}\;}}} \right)=\infty } \right)\\
 & \ \ \ \ \ \ \ \ \ \ \ \ \ \ \ \ \ \ \ \ \ \ \ \ \ \ \ \ \ \ \ =\underset{l\to \infty }{\mathop{\lim }}\,V\left( \sum\limits_{j=1}^{\infty }{{{g}_{{1}/{2}\;}}\left( \frac{\left| {{X}_{j}} \right|}{{{\left( Mj \right)}^{{1}/{p}\;}}} \right)>\frac{l}{2}} \right) \\
 & \ \ \ \ \ \ \ \ \ \ \ \ \ \ \ \ \ \ \ \ \ \ \ \ \ \ \ \ \ \ \ =\underset{l\to \infty }{\mathop{\lim }}\,\underset{n\to \infty }{\mathop{\lim }}\,V\left( \sum\limits_{j=1}^{\infty }{{{g}_{{1}/{2}\;}}\left( \frac{\left| {{X}_{j}} \right|}{{{\left( Mj \right)}^{{1}/{p}\;}}} \right)>\frac{l}{2}} \right)=1.
\end{align*}
On the other hand,\\
\begin{center}
\(\underset{n\to \infty }{\mathop{\lim \sup }}\,\frac{\left| {{X}_{n}} \right|}{{{n}^{{1}/{p}\;}}}\le \underset{n\to \infty }{\mathop{\lim \sup }}\,\left( \frac{\left| {{S}_{n}} \right|}{{{n}^{{1}/{p}\;}}}+\frac{\left| {{S}_{n-1}} \right|}{{{n}^{{1}/{p}\;}}} \right)\le 2\underset{n\to \infty }{\mathop{\lim \sup }}\,\frac{\left| {{S}_{n}} \right|}{{{n}^{{1}/{p}\;}}}.\)
\end{center}
It follows that\\
 \centerline{\(V\left( \underset{n\to \infty }{\mathop{\lim \sup }}\,\frac{\left| {{X}_{n}} \right|}{{{n}^{{1}/{p}\;}}}>m \right)=1,\forall m>0.\)}
This contradicts
 \(V\left( \underset{n\to \infty }{\mathop{\lim }}\,\frac{\left| {{X}_{n}} \right|}{{{n}^{{1}/{p}\;}}}=0 \right)=1.\)\\
 Therefore, the assumption \({{C}_{V}}\left( {{\left| {{X}_{1}} \right|}^{p}} \right)=\infty \) is incorrect.\\
 Finally, we have\\
 \centerline{\({{C}_{V}}\left( {{\left| {{X}_{1}} \right|}^{p}} \right)<\infty \).}

\end{proof}

Next, we will give a theorem about convergence of a random series
\begin{thm}\cite{zhang16c} Suppose that \(\left\{ {{X}_{i}} \right\}\) is a sequence of independent random variable in a sub-linear expectation space. If
\(\hat{E}\left( {{X}_{i}} \right)=\hat{\varepsilon }\left( {{X}_{i}} \right)=0\), \({{S}_{k}}=\sum\limits_{i=1}^{k}{{{X}_{i}}}\),
then\\
\centerline{\(\hat{E}\left( \underset{k\le n}{\mathop{\max }}\,{{\left| {{S}_{k}} \right|}^{p}} \right)\le {{C}_{p}}\left\{ \sum\limits_{k=1}^{n}{\hat{E}\left( {{\left| {{X}_{k}} \right|}^{p}} \right)+{{\left( \sum\limits_{k=1}^{n}{\hat{E}\left( {{\left| {{X}_{k}} \right|}^{2}} \right)} \right)}^{{p}/{2}\;}}} \right\}\).}\\
where \({{C}_{p}}\) is a positive constant depending only on \(p\).

\end{thm}

\begin{deft}\cite{sun15a}
A sub-linear expectation \(\hat{E}\) is called regular, if for any random variable sequence \(\left\{ {{X}_{n}} \right\}\) such that \({{X}_{n}}\downarrow 0\), we have\\
\centerline{\(\hat{E}\left( {{X}_{n}} \right)\downarrow 0\).}
\end{deft}
\begin{thm}\cite{sun15a}
If the sub-linear expecation \(\hat{E}\) is regular, and \(\left\{ {{\xi }_{n}} \right\}\) is a Cauchy sequence in capacity, there exists a subsequence \(\left\{ {{\xi }_{{{n}_{k}}}} \right\}\) converges to some \(\xi \) almost surely in capacity .
\end{thm}
\begin{thm}Suppose that  \(\left\{ {{X}_{i}} \right\}\) is a sequence of independent random variable. If \(\hat{E}\) is regular, \(\hat{E}\left( {{X}_{i}} \right)=\hat{\varepsilon }\left( {{X}_{i}} \right)=0\), and \(\sum\limits_{i=1}^{\infty }{\hat{E}\left( X_{i}^{2} \right)}<\infty \), then
\(\sum\limits_{i=1}^{\infty }{{{X}_{i}}}\) converges almost surely in capacity.
\end{thm}
\begin{proof}Define \({{S}_{n}}=\sum\limits_{i=1}^{n}{{{X}_{i}}}\). By the independence of random variable,
\(\forall \varepsilon >0\), if \(m > n\), then\\
\centerline{
\(V\left( \left| {{S}_{m}}-{{S}_{n}} \right|\ge \varepsilon  \right)\le \frac{1}{{{\varepsilon }^{2}}}\hat{E}\left( {{\left( {{S}_{m}}-{{S}_{n}} \right)}^{2}} \right)=\frac{1}{{{\varepsilon }^{2}}}\sum\limits_{k=n+1}^{m}{\hat{E}\left( X_{k}^{2} \right)}\to 0\), when \(m\to \infty \).}
So \(S_{n}\) is a Cauchy sequence in capacity.\\
Note that \(\hat{E}\) is regular,
so there exist a subsequence \(\left\{ {{n}_{k}} \right\}\) such that\\
\centerline{\({{S}_{{{n}_{k}}}}\to S\ a.s.\ V\).}
By Chebyshev inequality and theorem 8 with \(p=2\), we can get
\begin{align*}
  & \sum\limits_{k=1}^{\infty }{V\left( \underset{{{n}_{k}}<j\le {{n}_{k+1}}}{\mathop{\max }}\,\left| {{S}_{j}}-{{S}_{{{n}_{k}}}} \right|\ge \varepsilon  \right)}\le \sum\limits_{k=1}^{\infty }{\frac{\hat{E}\left( \underset{{{n}_{k}}<j\le {{n}_{k+1}}}{\mathop{\max }}\,{{\left| {{S}_{j}}-{{S}_{{{n}_{k}}}} \right|}^{2}} \right)}{{{\varepsilon }^{2}}}} \\
 & \ \ \ \ \ \ \ \ \ \ \ \ \ \ \ \ \ \ \ \ \ \ \ \ \ \ \ \ \ \ \ \ \ \ \ \ \ \ \ \ \ \ \ \ \le \frac{1}{{{\varepsilon }^{2}}}\sum\limits_{k=1}^{\infty }{2{{C}_{2}}}\sum\limits_{j={{n}_{k}+1}}^{{{n}_{k+1}}}{\hat{E}\left( {{\left| {{X}_{k}} \right|}^{2}} \right)} \\
 & \ \ \ \ \ \ \ \ \ \ \ \ \ \ \ \ \ \ \ \ \ \ \ \ \ \ \ \ \ \ \ \ \ \ \ \ \ \ \ \ \ \ \ \ \le \frac{2{{C}_{2}}}{{{\varepsilon }^{2}}}\sum\limits_{k=1}^{\infty }{\hat{E}\left( {{\left| {{X}_{k}} \right|}^{2}} \right).}
\end{align*}
By Borel-Cantelli lemma, we can get
\begin{center}
\(\underset{{{n}_{k}}<j\le {{n}_{k+1}}}{\mathop{\max }}\,\left| {{S}_{j}}-{{S}_{{{n}_{k}}}} \right|\to 0\ a.s.\ V.\)
\end{center}
Hence,\\
\centerline{\({{S}_{n}}\to S\ a.s.\ V\).}
\end{proof}
\textbf{Acknowledgements}: This work was supported by National Natural Science Foundation of China (Grant No. 11225104) and the Fundamental Research Funds for the Central Universities.

\end{document}